\documentclass[11pt]{amsart}

\usepackage{amssymb}
\usepackage{epsfig}
\usepackage{comment}
\usepackage{amsmath}
\usepackage{bm}
\usepackage{subcaption}
\usepackage[section]{placeins}
\usepackage{mathrsfs}
\usepackage{graphicx}
\usepackage[notrig]{physics}
\usepackage{units}
\usepackage{cleveref}           
\usepackage{upgreek} 
\usepackage{mathtools}
\usepackage{float}
\usepackage{color}

\newtheorem{thm}{Theorem}[section]

\newtheorem{lem}[thm]{Lemma}

\theoremstyle{definition}

\theoremstyle{remark}

\allowdisplaybreaks

\begin{document}

\title[Optimal layering]{
On the lowest-frequency bandgap\\ of 1D phononic crystals }

\author{
J.~Gonz\'{a}lez-Carbajal, 
M.~Lemm
and 
J.~Garcia-Suarez
}

\address
{
Javier Gonz\'{a}lez-Carbajal: Department of Mechanical and Manufacturing Engineering, University of Sevilla, Spain (current address: Mechanical Engineering, Rey Juan Carlos University, 28933 Móstoles, Madrid, Spain.) \\
Marius Lemm: Department of Mathematics, University of Tübingen, 72076 Tübingen, Germany. \\
Joaquin Garcia-Suarez: Institute of Civil Engineering, \'Ecole polytechnique fédérale de Lausanne (EPFL), 1015 Lausanne, Switzerland. 
}

\email{joaquin.garciasuarez@epfl.ch (corresponding author)}

\begin{abstract}
This manuscript puts forward and verifies an analytical  approach to design phononic crystals that feature a bandgap at the lowest possible frequencies. This new approach is verified against numerical optimization.
It rests on the exact form of the trace of the cumulative transfer matrix.
This matrix arises from the product of $N$ elementary transfer matrices, where $N$ represents the number of layers in the unit cell of the crystal. 
The paper presents first a new proof of the harmonic decomposition of the trace (which, unlike the original derivation, does not resort to group-theoretical concepts), 
and then goes to demonstrate that the long-wavelength asymptotics of the function that governs the dispersion relation can be described in closed form for any layering, plus that it possesses a simple and explicit form that facilitates its study. 
Using this asymptotic result for low frequencies, we then propose an analytical curvature-minimization approach to devise layerings that yield an opening of the first frequency bandgap at the lowest possible frequency value. 
Finally, we also show that minimizing the frequency at which the first gap opens can be at odds with obtaining broadband attenuation, i.e., maximizing the width of the first frequency bandgap.  

\vspace{1cm}

\noindent \textbf{keywords:} phononic crystal, laminate, wave propagation, bandgap, low frequency
\end{abstract}

\maketitle


\section{Introduction}

When it comes to acoustic and elastic wave control, remarkable achievements have been attained during the last decades in the field of metamaterials, see Refs.~\cite{Hussein:2014,Mace:2014,Review:2023} and references therein. 
The possibility of engineering the composition and/or structure of a complex medium in order to achieve exotic behaviors has been documented for, e.g., 
low-frequency bandgaps \cite{Wang:2022,Cao:2022}, surface-to-bulk conversion \cite{Colquitt:2017}, cloaking \cite{cloaking:2011,Amanat:2022}, broadband attenuation \cite{Hussein_design:2006}, wave-based logical gates \cite{Deymier:2011}, implementation of analog computations \cite{analog:2021}, to name a few. 
New numerical techniques beyond finite elements \cite{Orris:1974,Rohan:2009,Shim:2016} or finite differences \cite{Yamg:1974} are also being put forward for better understanding and design. 
By way of illustration, let us mention efficient FFT solvers \cite{Segurado:2023} and ad-hoc optimized discretization schemes \cite{Aragon:2021}, genetic algorithms \cite{Hamza:2006}, and ray tracing techniques borrowed from the seismology community \cite{Dorn:2022}. 
Despite these advances, design of 2D and 3D metamaterials still relies oftentimes on the physical intuition of the designer.

Metamaterials for elastic [acoustic] vibration control are interchangeably termed ``elastic [acoustic] metamaterials'' or ``phononic materials'' \cite{Hussein:2014,Delves:2014}. 
Those consisting of repetitive unit cells are termed ``phononic crystals''  \cite{Kushwaha:1993}. 
There is the implicit scale assumption that the material behaves as a continuum, i.e., atomistic effects do not enter the picture. 
The field of nanophononics \cite{nanophononics:2016,Wood:2023} is occupied with phonon propagation at the atomic scale, where both momentum and heat transport manifest simultaneously. 

In practice, a 1D phononic crystals can be an infinite periodic either laminate or rod \cite{Hussein_analysis:2006}, in which pressure (or decoupled shear waves) propagate perpendicularly to all interfaces. 
The scholar inquiry as to these systems dates back over half a century \cite{Hussein:2014}. 
One of the hurdles impeding deeper appraisals is the sheer mathematical complexity arising from considering anything beyond a couple of layers in the unit cells, 
as each new layer adds at least two mechanical parameters (a material stiffness plus density) and one geometrical one (thickness) for consideration. 
Moreover, those parameters appear into complex trigonometric functions \cite{Shen_&_Cao,Lustig:2018}.  
%
Hence, when considering more than two phases \cite{Hussein_design:2006}, it should not come as a surprise that design strategies based on sophisticated parametric analyses are gaining steam \cite{2019_parametric_study,2020_parametric_analysis}, neither should the fact that machine learning has been considered as a possibility to overcome these obstacles \cite{kaust_nn:2023,neural_network:2019,Luo:2020}.

Conversely, Ref.~\cite{JMPS} proposed an analytical approach founded on application  of transfer matrices.  
Even though its applicability extends to complex wave phenomena (including surface waves and mode conversion \cite{Aki-Richards}),
the transfer matrix formalism is particularly convenient for 1D wave propagation as it entails ``small'' 2-by-2 matrices.
This formalism has been shown to be an ideal means to understand wave propagation in layered media: 
each layer is represented by a matrix that depends on the frequency of the waves and on the material properties of the layer, 
the cumulative effect of subsequent layers is captured in a global matrix that is but the product of the individual ones. 
In particular, Bloch condition \cite{Kittel} for periodic laminates comes expressed in terms of the trace of the global matrix. 
This condition renders an equation (the ``dispersion relation'') relating the wavelength and the frequency within the laminate: 
combinations that yield real-valued wavelengths correspond to ``propagating'' waves (and the velocity at which they propagate), 
while imaginary solutions correspond to ``evanescent'' waves, whose amplitude decay exponentially fast. 
The former forms regions in the frequency domain usually termed ``passbands'', while the latter forms so-called ``stopbands'' or ``bandgaps'' (term of choice in this manuscript). 
Thus, vibration/acoustic insulation can be attained by means of the bandgaps in phononic crystals. 
%
%
In Ref.~\cite{JMPS}, it was proved that the trace, as a function of frequency, can be expressed in closed form, particularly as a superposition of cosine functions (harmonics), 
whose period and amplitude can be fully described in closed-form given a layering order and mechanical properties, for any number of layers. 

Revealing this structure allows new ways to understand performance and design of 1D phononic crystals, one that stems from ``force of reason'' instead of from sheer computational prowess \cite{2019_parametric_study,2020_parametric_analysis,kaust_nn:2023,neural_network:2019}. 
The exact general form of the entries of the global transfer matrix has already been used for applications in geotechnical earthquake engineering featuring layered soils \cite{SDEE:2022}, 
and, by analogy, could also be used for bandgap engineering in, e.g., photonics \cite{Kushwaha:1993}. 

We remark another recent (and complementary) key development in the field: the discovery of the representation on the torus (a finite geometric object) in which the whole band structure of laminates are encapsulated, independently of layering details \cite{Shmuel:2016}. 
By considering the trace of the global transfer matrix to be a function that defines a flow over the torus, exact expressions and bounds for the width of bandgaps were obtained, along with exact statistics regarding the ``gap density'', i.e., the proportion of the complete infinite spectrum that corresponds to bandgaps. 
%
%
Subsequently \cite{Lustig:2018}, some of the results were extended to more complex laminates and questions like ``can the first bandgap in a two-layer laminate be widened by adding a third layer to the unit cell that does not introduce a larger impedance contrast?'' (the answer in this particular case is yes). 
Refs.\,\cite{Shmuel:2016,Lustig:2018} proved that the exploiting toroidal representation allows for new analytical investigations of the problem that were not possible before. 
This technique has also been applied to a particular family quasi-crystalline 1D structures \cite{Gei:2010,Morini:2018,Gei:2020} (``Fibonacci structures''), achieving again remarkable results \cite{Morini:2019}, e.g., identifying the widest bandgap in the spectrum (wich does not have to be the first one necessarily) and widening of the lowest-frequency bandgap of the spectrum. 

In the text, we focus on deriving a criterion to ``optimize'' the first bandgap, in the sense of reducing  as much as possible the frequency at which it opens (we termed this the ``cut-off frequency''). 
It must be highlighted additionally that this contribution does not consider either that some layers may resonate \cite{Liu:2000,Baravelli:2014,Bertoldi:2015,Li:2021} nor that the material possesses an engineered microstructure \cite{Bertoldi:2015,brief_review:2022}, what has been an approach used in the past to attain low-frequency bandgaps. 
We will show later that this is equivalent to minimizing the first frequency for which the half-trace function equals minus one. The proposed criterion is based on a minimum-curvature condition and, very conveniently, we find that the optimal solution can be expressed in an explicit analytical manner.
We begin by providing a new algebraic derivation of the key result (exact general form of half-trace function and its harmonic decomposition), which was previously obtained resorting to group theory, i.e., identifying the transfer matrices as elements of the group $SL(2,\mathbb{R})$ \cite{Hall}. 
%
%


The text is structured as follows: 
\Cref{sec:review} presents the governing equations. 
\Cref{sec:analytical} covers the new proof 
as well as the asymptotic analysis for low frequencies, which leads to the optimization criterion. 
\Cref{sec:numerical_verif} contains three examples where the numerically-optimized layering is compared to the one derived from our criterion. 
We ponder the relation between lowering the frequency of the first bandgaps and maximizing its width in \Cref{sec:broadband}, and, by means of a counterexample, we show that designing for the former does not automatically carry out the latter. 
The manuscript closes with a summary and outlining potential future work in \Cref{sec:final}.

\section{Review}
\label{sec:review}

\subsection{The system}

\Cref{fig:laminate} depicts the system: a representative cell made of $N$ layers of different material and cross-sections, with thicknesses $l_i$, $i=1,...,N$, in an otherwise repetitive system assumed to be infinite.
 
\begin{figure}
\centering\includegraphics[width=0.99\textwidth]{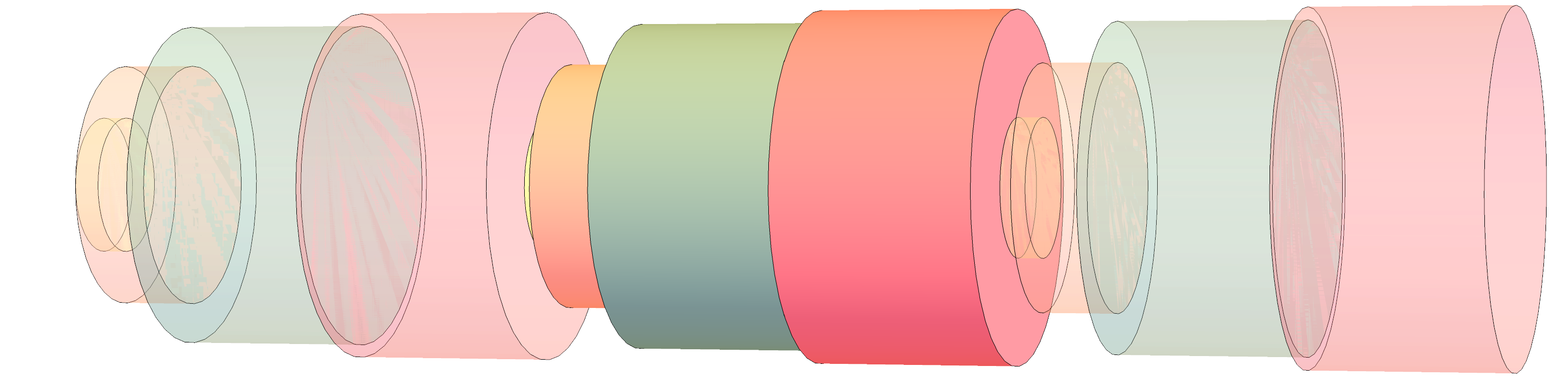}
\caption{Representative unit cell (highlighted) within a periodic rod where longitudinal waves propagate.}
\label{fig:laminate}
\end{figure}

The propagation of longitudinal waves in this system renders a 1D problem as the direction of propagation is perpendicular to all the interfaces between layers (and cross-section dispersive effects are not being considered \cite{Graff}). 
The layer heterogeneity is described by the position-dependence in the cell of the Young's modulus, and the cross-section area. 
Material properties and cross-section are assumed to remain constant in each layer. 
This setting is conceptually similar to the ones studied in Refs. \cite{Shmuel:2016,Lustig:2018}.


\subsection{Governing equations}

In this scenario (rod axial waves, ignoring dispersive effects), the governing equations, in time domain, are
\begin{subequations}
\begin{align}
    {\partial \over \partial x} (\sigma A) 
    =& 
    \rho {\partial^2 u \over \partial t^2} \, , 
    \\
    {\partial u \over \partial x}
    =&
    {\sigma \over E} \, ,
\end{align}
\end{subequations}
where $x$ is the longitudinal space coordinate, $t$ represents time, $\sigma$ is the normal stress to the interfaces, $u$ is the longitudinal displacement, $E$ is the Young modulus of the material the rod layer (``sub-rod'') is made of, $A$ is the cross-section area and $\rho$ must be understood as a linear density, i.e., mass per unit of length. 
Either by taking the Fourier transform or by assuming a plane-wave \textit{ansatz}, the frequency-domain version of these equations are
\begin{subequations}
\begin{align}
    {d \over d x} (\hat{\sigma} A) 
    =& 
    - \rho \omega^2 \hat{u} \, , \\
    {d \hat{u} \over d x}
    =&
    {\hat{\sigma} \over E} \, ,
\end{align}
\end{subequations}
where $\hat{u}$ and $\hat{\sigma}$ represent the amplitude of the displacement and the stress, respectively, and $\omega$ is the circular frequency. 
Recasting in matrix ODE form: for the $k$-th homogeneous portion \cite{Shmuel:2016},
\begin{align} \label{eq:layer_matrix_ODE_rods}
    { d \boldsymbol{f} 
    \over dx}
    = 
    { d \over dx}
    \begin{bmatrix}
    	\hat{u}  \\
        \hat{\sigma} A 
    \end{bmatrix}  
    = 
    \begin{bmatrix}
    0     &1/a_k \\
    - \rho_k \omega^2      &0
    \end{bmatrix}
    \begin{bmatrix}
    	\hat{u}  \\
        \hat{\sigma} A 
    \end{bmatrix}
    =
    \boldsymbol{A}_k \boldsymbol{f}(x)  
    \, ,
\end{align}
where $a_k = E_k A_k$ (do not confuse the $k$-th cross-section area, $A_k$ with the (boldface) $k$-th layer matrix $\boldsymbol{A}_k$). 
The solution, particularized at the end of the homogeneous sub-rod (whose length is $l_k$), is $\boldsymbol{f}(x_k) = \boldsymbol{T}_k (x_k,x_{k-1}) \boldsymbol{f}(x_{k-1})$, where
\begin{align} \label{eq:T_matrix}
    \boldsymbol{T}_k (x_k,x_{k-1})
    &= 
    \exp \left( l_k  
    \boldsymbol{A}_k
    \right) 
    =  
    \begin{bmatrix}
    \cos \left( {\omega l_k \over c_k }\right)     & {\sin \left( {\omega l_k \over c_k }\right) \over \omega \sqrt{\rho_k a_k}} \\
    - \omega \sqrt{\rho_k a_k} \, \sin \left( {\omega l_k \over c_k }\right)      &\cos \left( {\omega l_k \over c_k }\right)
    \end{bmatrix}
    \, , 
\end{align}
and $c_k=\sqrt{a_k/\rho_k}$ is the wave velocity in the $k$-th layer ($\mathrm{exp}$ represents matrix exponentiation).

By recursion, the motion-force vector $[ \hat{u}, \, \hat{\sigma} A ]^\top$ at any position in the cell (say, somewhere within the k-th layer) can be written in terms of the one at the initial point ($x=x_0$), in particular to the end of the $N$-th layer, $x_N$:
\begin{align}
    \boldsymbol{f}(x_N) 
    = 
    \boldsymbol{T}(x_N,x_0) \boldsymbol{f}(x_0) 
    = \boldsymbol{T}_N(x_N,x_{N-1}) 
    \ldots 
    T_k (x_k,x_{k-1}) \ldots T_1 (x_1,x_0)  
    \boldsymbol{f}(x_0) \, .
\end{align}

The matrix $\boldsymbol{T}$ is referred to as ``global transfer matrix'', ``cumulative transfer function'' \cite{Hussein_analysis:2006} or ``cumulative propagator''.


\section{Alternative derivation of entries of cumulative transfer function}
\label{sec:analytical}

\subsection{Harmonic decomposition of the trace}

We begin by presenting both a key result and a new purely-algebraic proof. 

\begin{thm}[Harmonic decomposition]
\label{thm:harmonic_decomposition}
Given any phononic crystal defined by a repetitive cell made up by $N$ layers stacked in a defined way, the trace of the corresponding cumulative transfer matrix can be expressed as
\begin{align} \label{eq:analytical_trace_body}
    \eta(\omega)
    =
    {1 \over 2} \mathrm{trace}(\boldsymbol{T}(L,0))
    =
    \sum_{k=1}^{2^{N-1}} 
    \mathcal{T}_{k} \cos\left( \tau_k \omega \right) \, ,
\end{align}
where the characteristic ``spectral periods'' are given by
\begin{align} \label{eq:periods}
    \tau_j
    =
    \left[
    {l_1 \over c_1}, \, 
    \ldots, \,
    {l_N \over c_N}
    \right]^{\top}
    \cdot
    \boldsymbol{\mathsf{e}}_j \, ,
\end{align}
the $j$-th permutation vector  $\boldsymbol{\mathsf{e}}_j \in \{ -1, \,1 \}^N$ can be any element in the set satisfying $\mathsf{e}_{j,1} = 1$, 
and the ``spectral amplitudes'' are
\begin{align} \label{eq:amplitudes}
    \mathcal{T}_j
    =
    {1 \over 2^{N-1}}
    \sum_{k = 0}^{\lfloor N /2 \rfloor}
    \sum_{|\mathsf{b}|=2 k }
    (-1)^{k+{\mathsf{b}\cdot \boldsymbol{\mathsf{e}}_j \over 2}}
    \upgamma_{\mathsf{b}} \, ,
\end{align}
the coefficient $\upgamma_{\mathsf{b}}$ associated with each binary multi-index \cite{evans} \textcolor{black}{$\mathsf{b} \in \{ 0, \,1 \}^N $} codifies the effect of the impedance contrast at the interfaces, and it is computed as
\begin{align} 
    \upgamma_{\mathsf{b}}
    =
    {1 \over 2}
    \left(
    Z^{\mathrm{f}(\mathsf{b})}
    +
    {1
    \over 
    Z^{\mathrm{f}(\mathsf{b})}}
    \right) \, ,    
\end{align}
where $Z_i = \omega\sqrt{\rho_i \, a_i}$ is the $i$-th layer impedance multiplied by frequency, 
and the map 
\textcolor{black}{
$\mathrm{f}
:
\{0,1\}^N 
\to 
\{-1,0,1\}^N$
} 
takes the multi-index $\mathsf{b}$, entry-wise, to the multi-index $\mathrm{f}(\mathsf{b})$ defined as: 
\begin{itemize}
    \item if $\mathsf{b}_i = 0$, then $\mathrm{f}(\mathsf{b}_i) = 0$,
    \item if $\mathsf{b}_i = 1$ and the previous value assigned by $\mathrm{f}$ to the prior 1-entry in $\mathsf{b}$ was $-1$, then $\mathrm{f}(\mathsf{b}_i) = 1$, else $\mathrm{f}(\mathsf{b}_i) = -1$.
    \item if $\mathsf{b}_i = 1$ and the previous value assigned by $\mathrm{f}$ to the prior 1-entry in $\mathsf{b}$ was $-1$, then $\mathrm{f}(\mathsf{b}_i) = 1$, else $\mathrm{f}(\mathsf{b}_i) = -1$.
\end{itemize}
\end{thm}

\begin{proof}
See that the $i$-th propagator can be written as
\begin{align}
    \boldsymbol{T}_i
    =
    \cos(r_i)
    \left(
    \boldsymbol{I}
    +
    \tan(r_i)
    \begin{bmatrix}
    0 &1/Z_i\\
    -Z_i &0
    \end{bmatrix}
    \right) \, ,
\end{align}
where $r_i=\omega l_i/c_i$, hence the cumulative propagator for a N-layer representative cell:
\begin{align} \label{eq:product_common_factor}
    \boldsymbol{T}
    =
    \prod_{i=1}^N
    \boldsymbol{T}_i
    =
    \prod_{i=1}^N
    \cos(r_i)
    \left(
    \boldsymbol{I}
    +
    \tan(r_i)
    \begin{bmatrix}
    0 &1/Z_i\\
    -Z_i &0
    \end{bmatrix}
    \right) \, .
\end{align}

See how the product of the cosine factors can be effectuated and passed to the left-hand side
\begin{subequations}
\begin{align}
    &\left(
    \prod_{i=1}^N \cos(r_i)
    \right)^{-1}
    \boldsymbol{T}
     =
    \prod_{i=1}^N
    \left(
    \boldsymbol{I}
    +
    \tan(r_i)
    \begin{bmatrix}
    0 &1/Z_i\\
    -Z_i &0
    \end{bmatrix}
    \right) 
    \shortintertext{thus we would have to multiply over all the terms that appear in the N binomials}
    &=\left(
    \boldsymbol{I}
    +
    \tan(r_1)
    \begin{bmatrix}
    0 &1/Z_1\\
    -Z_1 &0
    \end{bmatrix}
    \right)
    \ldots
    \left(
    \boldsymbol{I}
    +
    \tan(r_N)
    \begin{bmatrix}
    0 &1/Z_N\\
    -Z_N &0
    \end{bmatrix}
    \right)
    \shortintertext{this product can be expressed in compact form using multi-index notation, but some of the terms in the expansion are simple enough}
    &=
    \sum_{j=0}^{N}
    \sum_{|\mathsf{b}|=j}
    \tan(r)^\mathsf{b}
    \begin{bmatrix}
    0 &1/Z\\
    -Z &0
    \end{bmatrix}^\mathsf{b} \, , \\
    & =
    \boldsymbol{I}
    +
    \sum_{i=1}^{N}
    \tan(r_i)
    \begin{bmatrix}
    0 &1/Z_i\\
    -Z_i &0
    \end{bmatrix}
        \label{eq:proof_factors}
\\
    & \qquad \quad +
    \sum_{j=2}^{N-1}
    \sum_{|\mathsf{b}|=j}
    \tan(r)^\mathsf{b}
    \begin{bmatrix}
    0 &1/Z\\
    -Z &0
    \end{bmatrix}^\mathsf{b}
     \nonumber
    \\
    & \qquad \quad + 
    \prod_{i=1}^{N}
    \tan(r_i)
    \begin{bmatrix}
    0 &1/Z_i\\
    -Z_i &0
    \end{bmatrix}  \nonumber \, ,
\end{align}
\end{subequations}
where the binary multi-index $\mathsf{b} \in (\{0, \, 1\})^N$ is an $N$-tuple of numbers, each being either 0 or 1.  Therefore, the last term indicates taking the sum over all possible multi-indices of degree ($|\mathsf{b}|$) from 2 to $N-1$ (degree zero, one and N have been shown explicitly and correspond to the first three addends) and for each degree $j$, taking a sum over all possible binary multi-indices of such degree (there are $N! / (j! (N-j)!)$ binary multi-indices of degree $|\mathsf{b}|$, that is, whose N entries sum up to $|\mathsf{b}|$).
\subsection*{Diagonal elements}

Obviously, the first term (identity matrix) is a diagonal matrix, the second one cannot be diagonal, and the last one will be diagonal when $N$ is an even number. Concerning the two last terms, let us put them together in multi-index notation and then we can split into two pieces, one that will yield a diagonal matrix as result and another one that does not:
\begin{align} \label{eq:aux_odd_even}
\begin{split}
    \sum_{j=2}^{N}
    \sum_{|\mathsf{b}|=j}
    \tan(r)^\mathsf{b}
    \begin{bmatrix}
    0 &1/Z\\
    -Z &0
    \end{bmatrix}^\mathsf{b}
    &=
    \sum_{j \in \mathcal{I}_e}
    \sum_{|\mathsf{b}|=j}
    \tan(r)^\mathsf{b}
    \begin{bmatrix}
    0 &1/Z\\
    -Z &0
    \end{bmatrix}^\mathsf{b} \\
    &+
    \sum_{j \in \mathcal{I}_o}
    \sum_{|\mathsf{b}|=j}
    \tan(r)^\mathsf{b}
    \begin{bmatrix}
    0 &1/Z\\
    -Z &0
    \end{bmatrix}^\mathsf{b} \, ,
\end{split}
\end{align}
where $\mathcal{I}_e = \{ j \in [2,3,\dots,N] : j \text{ is even} \}$ $\mathcal{I}_o = \{ j \in [2,3,\dots,N] : j \text{ is odd} \}$. 
See that the set of values of $j$ to consider in this case can be also expressed as 
\begin{align}
    \mathcal{I}_e
    =
    \left[
    2,4,\ldots,
    2 \lfloor 
    N /2 
    \rfloor
    \right]
    =2
    \left[
    1,2,\ldots,
    \lfloor 
    N/2 
    \rfloor
    \right] \ , ,
\end{align}
the use of the floor function $\lfloor \cdot \rfloor$ accounts for the possibility of $N$ being itself even or odd. Likewise, the odd indices can be expressed as
\begin{align}
    \mathcal{I}_e
    =
    \left[
    1,3,\ldots,
    1+
    2 \lfloor 
    (N-1) /2 
    \rfloor
    \right]
    \ , .
\end{align}
Hence, the terms in \cref{eq:aux_odd_even} that yield a diagonal result can be expressed as
\begin{align}
    \sum_{j \in \mathcal{I}_e}
    \sum_{|\mathsf{b}|=j}
    \tan(r)^\mathsf{b}
    \begin{bmatrix}
    0 &1/Z_i\\
    -Z_i &0
    \end{bmatrix}^\mathsf{b}
    =
    \sum_{k = 1}^{\lfloor N/2 \rfloor}
    \sum_{|\mathsf{b}|=2 k }
    \tan(r)^\mathsf{b}
    \begin{bmatrix}
    0 &1/Z\\
    -Z &0
    \end{bmatrix}^\mathsf{b} \, .
\end{align}
For a given $\mathsf{b}$, $|\mathsf{b}|=2k$, take the $a$-th position to be the first non-zero entry, the $a$-th to be the second one, the $a'$-th one the prior to last and $b'$-th the last one. Hence, let us expand the matrix multiplication defined via this multi-index. Assume that the first 1-entry in $\mathsf{b}$ is in the $a$-th position, while the next 1 is in the $b$-th one, the second to last was in the $a'$-th and the last 1-entry corresponded to the $b'$-th entry:  
\begin{subequations}
\begin{align}
    \begin{bmatrix}
    0 &1/Z\\
    -Z &0
    \end{bmatrix}^\mathsf{b}
    & =
    \underbrace{
    \begin{bmatrix}
    0 &1/Z_a\\
    -Z_a &0
    \end{bmatrix}
    \begin{bmatrix}
    0 &1/Z_b\\
    -Z_b &0
    \end{bmatrix}
    }_{\text{first pair}}
    \ldots
    \underbrace{
    \begin{bmatrix}
    0 &1/Z_{a'}\\
    -Z_{a'} &0
    \end{bmatrix}
    \begin{bmatrix}
    0 &1/Z_{b'}\\
    -Z_{b'} &0
    \end{bmatrix}
    }_{\text{last pair}} \, ,
    \shortintertext{grouping pairs, }
    & =
    \underbrace{
    \begin{bmatrix}
    -Z_b/Z_a &0\\
    0 &-Z_a/Z_b
    \end{bmatrix}
    \ldots
    \begin{bmatrix}
    -Z_{b'}/Z_{a'} &0\\
    0 &-Z_{a'}/Z_{b'}
    \end{bmatrix}
    }_{\text{$k$ factors }} \, , 
    \shortintertext{using the notation presented in \cite{JMPS}}
    &=
    (-1)^{|\mathrm{b}|/2}
    \begin{bmatrix}
    Z^{-\mathrm{f}(\mathsf{b})}
    &0\\
    0 
    &Z^{\mathrm{f}(\mathsf{b})}
    \end{bmatrix} \, , \label{eq:diagonal_matrix_ready}
\end{align}
\end{subequations}
where \cite{JMPS} $\mathrm{f}:\{0,1\}^N \to \{-1,0,1\}^N$ takes the multi-index $\mathsf{b}$, entry-wise, to the multi-index $\mathrm{f}(\mathsf{b})$ defined as: 
 \begin{itemize}
        \item if $\mathsf{b}_i = 0$, then $\mathrm{f}(\mathsf{b}_i) = 0$,
        \item if $\mathsf{b}_i = 1$ and it is the first instance of an entry being equal to 1, then $\mathrm{f}(\mathsf{b}_i) = 1$.
        \item if $\mathsf{b}_i = 1$ and the previous value assigned by $\mathrm{f}$ to the prior 1-entry in $\mathsf{b}$ was $-1$, then $\mathrm{f}(\mathsf{b}_i) = 1$, else $\mathrm{f}(\mathsf{b}_i) = -1$.
    \end{itemize}

As an illustrating example, consider $\mathsf{b} = (0,1,0,1,1,0,0,1) \in (\{0,\,1\})^{8}$ (one of the combinations that would appear in a 8-layer laminate), thus $\mathrm{f}(\mathsf{b}) = (0,1,0,-1,1,0,0,-1)$, hence
\begin{align}
    \begin{bmatrix}
    0 &1/Z\\
    -Z &0
    \end{bmatrix}^\mathsf{b}
    =
    +
    \begin{bmatrix}
    { Z_4\,Z_8
    \over
    Z_2\,Z_5} &0 \\
    0 &{ Z_2\,Z_5
    \over
    Z_4\,Z_8}
    \end{bmatrix} \, .
\end{align}

\subsection*{Anti-diagonal elements}

We can also look into the anti-diagonal terms og \cref{eq:proof_factors}. In  similar fashion,
\begin{align}
    \sum_{j \in \mathcal{I}_o}
    \sum_{|\mathsf{b}|=j}
    \tan(r)^\mathsf{b}
    \begin{bmatrix}
    0 &1/Z\\
    -Z &0
    \end{bmatrix}^\mathsf{b}
    =
    \sum_{k = 1}^{\lfloor (N - 1)/2 \rfloor }
    \sum_{|\mathsf{b}|=2 k + 1 }
    \tan(r)^\mathsf{b}
    \begin{bmatrix}
    0 &1/Z\\
    -Z &0
    \end{bmatrix}^\mathsf{b} \, ,
\end{align}
so for a given $\mathsf{b}$ such that $|\mathsf{b}|=2k+1$ (taking advantage of the notation introduced in the prior case), 
\begin{subequations}
\begin{align}
    \begin{bmatrix}
    0 &1/Z\\
    -Z &0
    \end{bmatrix}^\mathsf{b}
    & =
    \underbrace{
    \begin{bmatrix}
    -Z_b/Z_a &0\\
    0 &-Z_a/Z_b
    \end{bmatrix}
    \ldots
    \begin{bmatrix}
    -Z_{b'}/Z_{a'} &0\\
    0 &-Z_{a'}/Z_{b'}
    \end{bmatrix}
    }_{\text{$k$ factors }}
    \begin{bmatrix}
    0 &1/Z_c\\
    -Z_c &0
    \end{bmatrix}
    \, , 
    \shortintertext{which is the result we got before bar one last factor}
    &=
    (-1)^{k}
    \begin{bmatrix}
    Z^{-\mathrm{f}(\mathsf{b})}
    &0\\
    0 
    &Z^{\mathrm{f}(\mathsf{b})}
    \end{bmatrix}
    \begin{bmatrix}
    0 &1/Z_c\\
    -Z_c &0
    \end{bmatrix} \, , \\
    &=
    \begin{bmatrix}
    0
    &{(-1)^{k} \over Z_c}
    Z^{-\mathrm{f}(\mathsf{b})}
    \\
    (-1)^{k+1} Z_c
    Z^{\mathrm{f}(\mathsf{b})}
    &0
    \end{bmatrix}
    \, .
\end{align}
\end{subequations}

Thus the addends corresponding to odd-degree multi-indices can simply be calculated as the immediately prior even-degree multi-index multiplied by the extra entry. 

\subsection*{Closed-form expression for the half-trace function}

Using the previous results, \cref{eq:diagonal_matrix_ready}, the half-trace function can be written as
\begin{align} \label{eq:eta_omega}
    \eta(\omega)
    =
    {1 \over 2}
    \mathrm{trace}
    \left(
    \boldsymbol{T}
    \right)
    =
    \left(
    \prod_{i=1}^N
    \cos(r_i)
    \right)
    \left(
    \sum_{k = 0}^{\lfloor N /2 \rfloor}
    \sum_{|\mathsf{b}|=2 k }
    (-1)^{|\mathsf{b}|/2}
    \gamma_{\mathsf{b}}
    \tan(r)^\mathsf{b}
    \right) \, ,
\end{align}
where the coefficient $\gamma_{\mathsf{b}}$ is, for each multi-index $\mathsf{b}$ and corresponding tangent factor,
\begin{align} \label{eq:Z_b}
    \gamma_{\mathsf{b}}
    =
    {1 \over 2}
    \left(
    Z^{\mathrm{f}(\mathsf{b})}
    +
    Z^{-\mathrm{f}(\mathsf{b})}
    \right) \, ,
\end{align}
and it is computed easily from the rules outlined above. Performing some algebraic manipulations over \cref{eq:eta_omega}, we find
\begin{subequations}
\begin{align}
    \eta(\omega)
    & =
    \left(
    \prod_{i=1}^N
    \cos(r_i)
    \right)
    \left(
    \sum_{k = 0}^{\lfloor N /2 \rfloor}
    \sum_{|\mathsf{b}|=2 k }
    (-1)^{|\mathsf{b}|/2}
    \gamma_{\mathsf{b}}
    \tan(r)^\mathsf{b}
    \right) \, , \\
    &=
    \sum_{k = 0}^{\lfloor N /2 \rfloor}
    \sum_{|\mathsf{b}|=2 k }
    (-1)^{|\mathsf{b}|/2}
    \gamma_{\mathsf{b}}
    \cos(r^\mathsf{b}) \, ,
\end{align}
\end{subequations}
where the new notation $\cos(r^\mathsf{b})$ must be explained before continuing the derivation: as we multiply the cosine factors by each element of the sum that makes up the second factor, some of those cosines will become sines when combined with the tangent factors; however, we can think of the sines as cosines with extra phase $-\pi/2$, thus
\begin{subequations}
\begin{align}
    \cos(r^\mathsf{b})
    &=
    \prod_{i=1}^{N}
    \cos
    \left(
    r_i - \mathsf{b}_i {\pi \over 2}
    \right) \, ,
    \intertext{and can also apply a well-known trigonometric identity termed ``product to sum'' formula so as to reach}
    &=
    {1 \over 2^N}
    \sum_{\boldsymbol{\mathsf{e}}' \in \{-1,1\}^{N}}
    (-1)^{\mathsf{b}\cdot \boldsymbol{\mathsf{e}}'}
    \cos(\boldsymbol{r} \cdot \boldsymbol{\mathsf{e}}') \, ,
    \intertext{where $\boldsymbol{r}$ is just the vector encompassing all $r_i$, from $i=1$ to $i=N$, while $\boldsymbol{\mathsf{e}}'$ is a new binary multi-index, $N$ entries, this time taking $\pm 1$ values; since $\boldsymbol{\mathsf{e}}'$ appears in the argument of a cosine and $|\mathsf{b}|$ is even, each term in the sum appears exactly twice and thus we can restrict the sum to the elements $\boldsymbol{\mathsf{e}} \in \{-1,1\}^{N}$ s.t. $e_1 = 1$, so}
    &=
    {1 \over 2^{N-1}}
    \sum_{\boldsymbol{\mathsf{e}} \in \{-1,1\}^{N}}
    (-1)^{\mathsf{b}\cdot \boldsymbol{\mathsf{e}}}
    \cos(\boldsymbol{r} \cdot \boldsymbol{\mathsf{e}}) \, ,
\end{align} 
\end{subequations}
so, for each multi-index $\mathsf{b}$, the trigonometric factors can be turned into a sum of $2^{N-1}$ cosines terms, whose arguments are combinatorially defined in terms of the $r_i$, always the same ones, the only thing changing from one multi-index to another being the sign with which the cosines appear in the sum.

Plugging the later result in \cref{eq:eta_omega},
\begin{subequations}
\begin{align}
    \eta(\omega)
    &=
    \sum_{k = 0}^{\lfloor N /2 \rfloor}
    \sum_{|\mathsf{b}|=2 k }
    (-1)^{|\mathsf{b}|/2}
    \gamma_{\mathsf{b}}
    \cos(r^\mathsf{b}) \, , \\
    &=
    \sum_{k = 0}^{\lfloor N /2 \rfloor}
    \sum_{|\mathsf{b}|=2 k }
    {(-1)^{|\mathsf{b}|/2} \gamma_{\mathsf{b}} 
    \over 
    2^{N-1}}
    \sum_{\boldsymbol{\mathsf{e}} \in (\{-1,1\})^{N}}
    (-1)^{\mathsf{b}\cdot \boldsymbol{\mathsf{e}}}
    \cos(\boldsymbol{r} \cdot \boldsymbol{\mathsf{e}}) \, ,
    \shortintertext{re-arranging the sums,}
    &=
    \sum_{\boldsymbol{\mathsf{e}} \in \{-1,1\}^{N}}
    \left(
    {1 \over 2^{N-1}}
    \sum_{k = 0}^{\lfloor N /2 \rfloor}
    \sum_{|\mathsf{b}|=2 k }
    (-1)^{|\mathsf{b}|/2+\mathsf{b}\cdot \boldsymbol{\mathsf{e}}}
    \gamma_{\mathsf{b}}
    \right)
    \cos(\boldsymbol{r} \cdot \boldsymbol{\mathsf{e}}) 
     \, , 
     \shortintertext{so instead of summing over the values of the multi-index $\mathsf{b}$ last, we do it first, and lastly over $\boldsymbol{\mathsf{e}}$,}
     & =
     \sum_{j=1}^{2^{N-1}}
     \mathcal{T}_j
     \cos(\tau_j \omega) \, ,
\end{align}
\end{subequations}
thus we reach a harmonic decomposition of the half-trace function. The $j$-th ``harmonic period'' $\tau_j$ can be computed as
\begin{align}
    \tau_j
    =
    \left[
    {l_1 \over c_1}, \, 
    \ldots, \,
    {l_N \over c_N}
    \right]^{\top}
    \cdot
    \boldsymbol{\mathsf{e}}_j \, ,
\end{align}
while the corresponding ``harmonic amplitude'',
\begin{align}
    \mathcal{T}_j
    =
    {1 \over 2^{N-1}}
    \sum_{k = 0}^{\lfloor N /2 \rfloor}
    \sum_{|\mathsf{b}|=2 k }
    (-1)^{|\mathsf{b}|/2+\mathsf{b}\cdot \boldsymbol{\mathsf{e}}}
    \gamma_{\mathsf{b}} \, ,
\end{align}
the coefficient associated to each multi-index was computed via \cref{eq:Z_b}.
\end{proof}


\subsection{The curvature at the origin}
\label{sec:curvature}

This result is presented as a consequence of \cref{thm:harmonic_decomposition}. Recall that  $Z_i = \omega \sqrt{ \rho_i a_i }$ and set $t_i = l_i / c_i$.

\begin{lem}[Curvature]
For any layering, the half-trace function features zero slope at $\omega = 0$, and negative curvature given by
\begin{align} 
\label{eq:kappa}
    \kappa
    =
    \sum_{i = 1}^{N}
    \sum_{j = 1}^{N}
    \left(
    Z_i \over Z_j
    \right)
    t_i t_j
    =
    (\boldsymbol{t} \cdot \boldsymbol{Z})
    (\boldsymbol{t} \cdot \boldsymbol{Z}^{-1})
    \, ,
\end{align}
with the vectors $\boldsymbol{t}=\left(t_1,t_2,\dots,t_N\right)$, $\boldsymbol{Z}=\left(Z_1,Z_2,\dots,Z_N\right)$ and $\boldsymbol{Z}^{-1}=\left(1/Z_1,1/Z_2,\dots,1/Z_N\right)$. 
\end{lem}
Alternatively, the curvature can also be expressed (in terms of design parameters) as
\begin{equation}
\label{eq:kappa_3}
\kappa=\left(\bm{l}\cdot\bm{\rho}\right)\left(\boldsymbol{l}\cdot\bm{a}^{-1}\right)
\end{equation}
where $\bm{l}=\left(l_1,l_2,\dots,l_N\right)$, $\bm{\rho}=\left(\rho_1,\rho_2,\dots,\rho_N\right)$ and $\bm{a}=\left(a_1,a_2,\dots,a_N\right)$.

We note that Eq.~\eqref{eq:kappa_3} is one of the main contributions of the present work. 
This relation constitutes a simple and direct way to calculate the curvature of the half-trace function at the origin, given the values of thickness, density and stiffness for the different components of the unit cell. As we show in the next section, this makes it possible to analytically minimize the curvature, which is at the heart of our novel approach to bandgap design.

\begin{proof}
Starting from \cref{eq:product_common_factor}, let us compute the derivatives of the transfer function at the origin, i.e., $\omega = 0$ or, equivalently, $r_i = 0$ $\forall i = 1, \ldots , N$: 
\begin{subequations}
\begin{align}
    { d^2 \over d \omega^2}
    \boldsymbol{T}
    \Big|_{\omega = 0}
    & =
    { d^2 \over d \omega^2}
    \left[
    \prod_{i=1}^N
    \cos(r_i)
    \prod_{j=1}^N
    \left(
    \boldsymbol{I}
    +
    \tan(r_j)
    \begin{bmatrix}
    0 &1/Z_j\\
    -Z_j &0
    \end{bmatrix}
    \right)
    \right]_{\omega = 0}
    \, , 
    \shortintertext{the product is split in two parts, indexed by $i$ and $j$ }
     = & 
    \left[
    { d^2 \over d \omega^2}
    \prod_{i=1}^N
    \cos(r_i)
    \right]_{\omega = 0}
    \left[
    \prod_{j=1}^N
    \left(
    \boldsymbol{I}
    +
    \tan(r_j)
    \begin{bmatrix}
    0 &1/Z_j\\
    -Z_j &0
    \end{bmatrix}
    \right)
    \right]_{\omega = 0} \\
    & +
    2
    \left[
    {d \over d \omega}
    \prod_{i=1}^N
    \cos(r_i)
    \right]_{\omega = 0}
    \left[
    {d \over d \omega}
    \prod_{j=1}^N
    \left(
    \boldsymbol{I}
    +
    \tan(r_j)
    \begin{bmatrix}
    0 &1/Z_j\\
    -Z_j &0
    \end{bmatrix}
    \right)
    \right]_{\omega = 0} \\
    & +
    \left[
    \prod_{i=1}^N
    \cos(r_i)
    \right]_{\omega = 0}
    \left[
    { d^2 \over d \omega^2}
    \prod_{j=1}^N
    \left(
    \boldsymbol{I}
    +
    \tan(r_j)
    \begin{bmatrix}
    0 &1/Z_j\\
    -Z_j &0
    \end{bmatrix}
    \right)
    \right]_{\omega = 0} \, ,
    \shortintertext{the cross-terms vanish since there always appear a term $\sin (0)$ in the first factor, while the second factor of the first addend and the first factor of the second yield the identity matrix and 1 respectively}
     = & 
    \boldsymbol{I}
    \left[
    { d^2 \over d \omega^2}
    \prod_{i=1}^N
    \cos(r_i)
    \right]_{\omega = 0}
    +
    \left[
    { d^2 \over d \omega^2}
    \prod_{j=1}^N
    \left(
    \boldsymbol{I}
    +
    \tan(r_j)
    \begin{bmatrix}
    0 &1/Z_j\\
    -Z_j &0
    \end{bmatrix}
    \right)
    \right]_{\omega = 0} \, .
\end{align}    
\end{subequations}

The chain rule can be brought to bear again. For the first addend, compute the first derivative first:
\begin{align}
    { d \over d \omega}
    \prod_{i=1}^N
    \cos(r_i)
    =
    - \sum_{i=1}^{N}
    t_i \sin(r_i)
    \left(
    \prod_{j \ne i}
    \cos(r_j)
    \right) \, ,
\end{align}
where $t_i=l_i/c_i$ . Thus, the second derivative at the origin
\begin{subequations}
\begin{align}
    { d^2 \over d \omega^2 }
    \prod_{i=1}^N
    \cos(r_i)
    \Big|_{\omega=0}
    & =
    - \sum_{i=1}^{N}
    {d \over d \omega}
    \left[
    t_i \sin(r_i)
    \prod_{j \ne i}
    \cos(r_j)
    \right]_{\omega = 0}  \, , \\
    & =
    - \sum_{i=1}^{N}
    \Bigg(
    \left[
    {d \over d \omega}
    t_i \sin(r_i)
    \right]_{\omega = 0}
    \left[
    \prod_{j \ne i}
    \cos(r_j)
    \right]_{\omega = 0}
    \nonumber \\ 
     &  \qquad \qquad \qquad \qquad + 
    \left[
    t_i \sin(r_i)
    \right]_{\omega = 0}
    \left[
    {d \over d \omega}
    \prod_{j \ne i}
    \cos(r_j)
    \right]_{\omega = 0}
    \Bigg) \, , 
    \shortintertext{evaluating at $\omega = 0$,}
    &=
    - \sum_{i=1}^{N} t_i^2 = - \Vert\boldsymbol{t}\Vert^{2} \, .
\end{align}    
\end{subequations}
Likewise, the first derivative of the second addend
\begin{align}
    &{d \over d \omega}
    \prod_{i=1}^N
    \left(
    \boldsymbol{I}
    +
    \tan(r_i)
    \begin{bmatrix}
    0 &1/Z_i\\
    -Z_i &0
    \end{bmatrix}
    \right) \nonumber \\
    =
    &\sum_{i=1}^{M}
    {t_i \over \cos^2(r_i)}
    \begin{bmatrix}
    0 &1/Z_i\\
    -Z_i &0
    \end{bmatrix}
    \prod_{i \ne j}
    \left(
    \boldsymbol{I}
    +
    \tan(r_j)
    \begin{bmatrix}
    0 &1/Z_j\\
    -Z_j &0
    \end{bmatrix}
    \right) \, ,
\end{align}
while the second one
\begin{subequations}
\begin{align}
    & {d \over d \omega}
    \left[
    \sum_{i=1}^{N}
    \begin{bmatrix}
    0 &1/Z_i\\
    -Z_i &0
    \end{bmatrix}
    {t_i \over \cos^2(r_i)}
    \prod_{i \ne j}
    \left(
    \boldsymbol{I}
    +
    \tan(r_j)
    \begin{bmatrix}
    0 &1/Z_j\\
    -Z_j &0
    \end{bmatrix}
    \right)
    \right]_{\omega = 0} 
    \, , \\
    & = 
    \sum_{i=1}^{N}
    {d \over d \omega}
    \left[
    \begin{bmatrix}
    0 &1/Z_i\\
    -Z_i &0
    \end{bmatrix}
    {t_i \over \cos^2(r_i)}
    \prod_{i \ne j}
    \left(
    \boldsymbol{I}
    +
    \tan(r_j)
    \begin{bmatrix}
    0 &1/Z_j\\
    -Z_j &0
    \end{bmatrix}
    \right)
    \right]_{\omega = 0} 
    \, , \\
    & = 
    \sum_{i=1}^{N}
    \begin{bmatrix}
    0 &1/Z_i\\
    -Z_i &0
    \end{bmatrix}
    {d \over d \omega}
    \left[
    {t_i \over \cos^2(r_i)}
    \right]_{\omega = 0}
    \prod_{i \ne j}
    \left(
    \boldsymbol{I}
    +
    \tan(r_j)
    \begin{bmatrix}
    0 &1/Z_j\\
    -Z_j &0
    \end{bmatrix}
    \right) \nonumber \\
    & + 
    \sum_{i=1}^{N}
    {t_i \over \cos^2(r_i)}
    {d \over d \omega}
    \left[
    \prod_{i \ne j}
    \left(
    \boldsymbol{I}
    +
    \tan(r_j)
    \begin{bmatrix}
    0 &1/Z_j\\
    -Z_j &0
    \end{bmatrix}
    \right)
    \right]_{\omega = 0}
    \, , \\
    & = 
    \sum_{i=1}^{N}
    \begin{bmatrix}
    0 &1/Z_i\\
    -Z_i &0
    \end{bmatrix}
    \begin{bmatrix}
    0 &1/Z_i\\
    -Z_i &0
    \end{bmatrix}
    \left[
    {2 t_i \tan(r_i) \over \cos^2(r_i)}
    \right]_{\omega = 0}
    \prod_{i \ne j}
    \left(
    \boldsymbol{I}
    +
    \tan(r_j)
    \begin{bmatrix}
    0 &1/Z_j\\
    -Z_j &0
    \end{bmatrix}
    \right) \nonumber \\
    & + 
    \sum_{i=1}^{N}
    \begin{bmatrix}
    0 &1/Z_i\\
    -Z_i &0
    \end{bmatrix}
    {t_i \over \cos^2(r_i)}
    {d \over d \omega}
    \left[
    \prod_{i \ne j}
    \left(
    \boldsymbol{I}
    +
    \tan(r_j)
    \begin{bmatrix}
    0 &1/Z_j\\
    -Z_j &0
    \end{bmatrix}
    \right)
    \right]_{\omega = 0}
    \, , \\
    & = 
    \sum_{i=1}^{N}
    \begin{bmatrix}
    0 &1/Z_i\\
    -Z_i &0
    \end{bmatrix}
    \left[
    {t_i \over \cos^2(r_i)}
    \right]_{\omega = 0} \noindent \\
    & \qquad \qquad
    \sum_{j=1}^{N}
    \begin{bmatrix}
    0 &1/Z_j\\
    -Z_j &0
    \end{bmatrix}
    \left[
    {t_j \over \cos^2(r_j)}
    \right]_{\omega = 0}
    \left[
    \prod_{i \ne k \ne j}
    \left(
    \boldsymbol{I}
    +
    \tan(r_k)
    \begin{bmatrix}
    0 &1/Z_k\\
    -Z_k &0
    \end{bmatrix}
    \right)
    \right]_{\omega = 0}
    \, , \\
    & = 
    \sum_{i=1}^{N}
    \begin{bmatrix}
    0 &1/Z_i\\
    -Z_i &0
    \end{bmatrix}
    t_i
    \sum_{j=1}^{M}
    \begin{bmatrix}
    0 &1/Z_j\\
    -Z_j &0
    \end{bmatrix}
    t_j \, , \\
    & =
    2
    \sum_{i=1}^{N}
    \begin{bmatrix}
    0 &1/Z_i\\
    -Z_i &0
    \end{bmatrix}
    \begin{bmatrix}
    0 &1/Z_j\\
    -Z_j &0
    \end{bmatrix}
    t_i
    t_j 
    = 
    2
    \sum_{i=1}^{N}
    \begin{bmatrix}
    -Z_j/Z_i &0\\
    0 &-Z_i/Z_j
    \end{bmatrix}
    t_i
    t_j 
    \, .
\end{align}    
\end{subequations}


\begin{align}
\label{eq_sec_der_trace}
    { d^2 \over d \omega^2}
    \mathrm{trace}
    \left( \boldsymbol{T} \right)
    \Big|_{\omega = 0}
    =
    - 2 
    \sum_{i=1}^{N}
    t_i^2
    -
    2 
    \sum_{i = 1}^{N}
    \sum_{j \ne i}
    \left(
    Z_i \over Z_j
    \right)
    t_i t_j \, .
\end{align}
the factor of two in the first addend comes from the trace of the identity matrix.

The half-trace function can be expanded around $\omega = 0$ as
\begin{align}
\label{eq_half_trace_taylor}
    \eta(\omega)={1 \over 2}
    \mathrm{trace}
    \left( \boldsymbol{T} \right)
    =
    1 
    - 
    { \kappa \over 2 } \omega^2 
    +
    \mathcal{O}(\omega^4) \, ,
\end{align}
where it has been used that $\eta(0)=1$ (the cumulative transfer matrix at zero frequency, $\boldsymbol{T}|_{\omega = 0}$, has to be equal to the identity matrix~\cite{JMPS}) and that all odd derivatives at the origin vanish ($\eta'(0)=\eta'''(0)=...=0$), since $\eta(\omega)$ is a sum of cosine functions, all of which have vanishing odd derivatives at the origin (see Eq.~\eqref{eq:analytical_trace_body}). By comparing Eqs.~\eqref{eq_sec_der_trace} and~\eqref{eq_half_trace_taylor}, it is clear that the curvature $\kappa$ is given by

\begin{align}
    \kappa
    =
    \sum_{i=1}^{N}
    t_i^2
    +
    \sum_{i = 1}^{N}
    \sum_{j \ne i}
    \left(
    Z_i \over Z_j
    \right)
    t_i t_j
    \label{eq:kappa_0}
    \, ,
\end{align}
where $t_i = l_i / c_i$, $l_i$ is the length of the $i-th$ layer and $c_i = \sqrt{a_i / \rho_i}$ the wave velocity in the material. The time that takes for a wave to traverse each layer can be recast in a vector:
\begin{align*}
    \boldsymbol{t}
    =
    \left(
    {l_1 \over c_1} \, ,
    {l_2 \over c_2} \, ,
    \ldots \, ,
    {l_N \over c_N}
    \right) \, .
\end{align*}

The first addend of \cref{eq:kappa_0} is clearly the $L^2$ norm, but when it comes to the second one is not so straightforward. 
However, see that the two terms can be combined into a single sum
\begin{align} 
    \label{eq:kappa_proof}
    \kappa
    =
    \sum_{i = 1}^{N}
    \sum_{j = 1}^{N}
    \left(
    Z_i \over Z_j
    \right)
    t_i t_j
    =
    (\boldsymbol{t} \cdot \boldsymbol{Z})
    (\boldsymbol{t} \cdot \boldsymbol{Z}^{-1})
    \, ,
\end{align}
with $\boldsymbol{Z}=\left(Z_1,Z_2,\dots,Z_N\right)$ and $\boldsymbol{Z}^{-1}=\left(1/Z_1,1/Z_2,\dots,1/Z_N\right)$. 
Using relations $Z_i = \omega \sqrt{\rho_i a_i}$ and $t_i=l_i/c_i$, Eq.~\eqref{eq:kappa_proof} can be rewritten as
\begin{equation}
\label{eq:kappa_3_proof}
\kappa=\left(\bm{l}\cdot\bm{\rho}\right)\left(\boldsymbol{l}\cdot\bm{a}^{-1}\right)
\end{equation}
where $\bm{l}=\left(l_1,l_2,\dots,l_N\right)$, $\bm{\rho}=\left(\rho_1,\rho_2,\dots,\rho_N\right)$ and $\bm{a}=\left(a_1,a_2,\dots,a_N\right)$. 

\end{proof}

\section{Optimization for given materials}

As was described in the introduction, the goal of this work is to design 1D phononic crystals that feature their first frequency bandgap at the lowest possible frequencies. The bandgaps corresponds to the frequency ranges for which $\vert\eta(\omega)\vert>1$~\cite{JMPS}. 

For the desired optimization, we can leverage the specific shape of the half-trace function near the origin. In particular, consider the following three properties of $\eta(\omega)$, which follow from Eqs.~\eqref{eq_half_trace_taylor} and~\eqref{eq:kappa_3_proof}: $\eta(0)=1$, $\eta'(0)=0$, $\eta''(0)=-\kappa<0$. In other words, the half-trace function starts at 1, with a zero slope and a negative curvature. Hence, in order to minimize the first frequency for which $\eta(\omega)=-1$ (first cut-off frequency), it seems reasonable to maximize the absolute value of the curvature, $\kappa$. Of course, this will not provide an exact solution to the minimization of the first cut-off frequency, since higher order terms are not being considered. However, it will be shown that this approach can provide an accurate approximation to the desired optimum.

It is convenient at this point to rewrite Eq.~\eqref{eq:kappa_3} separating the magnitudes and directions of the different vectors involved:
\begin{equation}
\label{eq:kappa_4}
\kappa=\Vert\bm{l}\Vert^2\Vert\bm{\rho}\Vert\Vert\bm{a}^{-1}\Vert
(\hat{\bm{l}}\cdot\hat{\bm{\rho}})(\hat{\boldsymbol{l}}\cdot\bm{\hat{a}}^{-1}),
\end{equation}
where normalized versions of vectors $\bm{l}$, $\bm{\rho}$ and $\bm{a}^{-1}$ have been introduced: $\hat{\bm{l}}=\bm{l}/\Vert\bm{l}\Vert$, $\hat{\bm{\rho}}=\bm{\rho}/\Vert\bm{\rho}\Vert$, $\bm{\hat{a}}^{-1}=\bm{a}^{-1}/\Vert\bm{a}^{-1}\Vert$.

Fixing $\Vert\boldsymbol{\rho}\Vert$ and $\Vert\boldsymbol{a}^{-1}\Vert$ amounts to deciding on the material of the layers while leaving, on one hand, $\boldsymbol{l}$, i.e., the thickness of these layers, and on the other hand, the ordering of the layers (the order of the entries of $\boldsymbol{\rho}$ and $\boldsymbol{a}^{-1}$) as the variables to maximize over. Choosing the ordering is equivalent to fully defining the vectors $\boldsymbol{a}^{-1}$ and $\boldsymbol{\rho}$.

Assuming that both the materials and the layering order are given, we are interested in finding the thickness values (components of vector $\bm{l}$) that maximize the absolute value of the curvature, $\kappa$. It is clear from Eq.~\eqref{eq:kappa_4} that, if no restriction is imposed on $\Vert\bm{l}\Vert$, the curvature could grow unbounded by just increasing the norm of vector $\bm{l}$, i.e. by increasing the total thickness of the repetitive cell of the crystal. However, in real applications there is always some practical constraint on the maximum allowable size of the crystal. For this reason, it is reasonable to include in the optimization process a restriction such as $\Vert\bm{l}\Vert\leq\Vert\bm{l}\Vert_{max}$. Moreover, the structure of Eq.~\eqref{eq:kappa_4} implies that, under the condition $\Vert\bm{l}\Vert\leq\Vert\bm{l}\Vert_{max}$, the minimum curvature necessarily occurs for $\Vert\bm{l}\Vert=\Vert\bm{l}\Vert_{max}$. 

According to the considerations of the last paragraph, the optimization problem of interest can be stated as: 
\textit{given a set of layers of different materials in a prescribed order 
(i.e., given $\boldsymbol{a}^{-1}$ and $\boldsymbol{\rho}$), 
and a condition as to the maximum length of the unit cell 
(i.e., $\Vert\bm{l}\Vert=\Vert\bm{l}\Vert_{max}$), 
find the thickness for each layer that minimizes the curvature of the trace function at zero frequency 
(i.e. find the unitary vector $\hat{\bm{l}}$ that maximizes $\kappa$, as specified by Eq.~\eqref{eq:kappa_4}).} 

The rest of this subsection is devoted to the derivation of the exact solution to this problem. We present the conclusion as a theorem, which we subsequently prove.

\begin{thm}
    Define the unit cell of phononic crystals up to the thickness of each layer (i.e., $\boldsymbol{a}$ and $\boldsymbol{\rho}$ known, $\boldsymbol{l}$ unknown), 
    provided that the total length satisfies $\norm{\boldsymbol{l}} = \norm{\boldsymbol{l}}_{max}$ (where the later length is specified); 
    the distribution of thicknesses that minimizes the curvature of the half-trace function at zero frequency is $\boldsymbol{l}^* = \norm{\boldsymbol{l}}_{max} \hat{\boldsymbol{l}}^*$, such that
    \begin{align}
\label{eq_analyt_opt_thm}
    \hat{\boldsymbol{l}}^*
    \parallel
    {1 \over 2}
    \left(
    \hat{\boldsymbol{\rho}}
    +
    \boldsymbol{\hat{a}}^{-1}
    \right) \, .
\end{align}
\end{thm}

\begin{proof}
Since the points pronged by $\hat{\boldsymbol{l}} \in \mathbb{R}^N$ live on the unitary $(N-1)$-sphere $S^{N-1}$, we can pose the problem as
\begin{align}
\label{eq:opt}
    \max_{ \hat{\boldsymbol{l}} \in S^{N-1} } \,
    (\hat{\boldsymbol{l}} \cdot \hat{\boldsymbol{\rho}})
    (\hat{\boldsymbol{l}} \cdot \boldsymbol{\hat{a}}^{-1})
    \, ,
\end{align}
which can be geometrically interpreted as maximizing the projection of $\hat{\boldsymbol{l}}$ over $\hat{\boldsymbol{\rho}}$ and $\hat{\boldsymbol{a}}^{-1}$ simultaneously. 

To maximize this projection, the first obvious step is to make sure that $\hat{\boldsymbol{l}}$ lies in the hyper-plane spanned by $\hat{\boldsymbol{\rho}}$ and $\boldsymbol{\hat{a}}^{-1}$. Let us compute the cosine of the angle between $\hat{\boldsymbol{\rho}}$ and $\boldsymbol{\hat{a}}^{-1}$, termed $\theta$:
\begin{align}
    \cos \theta
    =
    \hat{\boldsymbol{\rho}}
    \cdot
    \boldsymbol{\hat{a}}^{-1}
    =
    H\, , \label{eq:angle}
\end{align}
where the inner product $\hat{\boldsymbol{\rho}}\cdot\boldsymbol{\hat{a}}^{-1}$ has been denoted as $H$ for convenience.

It is useful to find an orthonormal basis, formed by vectors $\hat{\boldsymbol{w}}_1$ and $\hat{\boldsymbol{w}}_2$, that spans the same plane as $\boldsymbol{\hat{\rho}}$ and $\hat{\boldsymbol{a}}^{-1}$, using Gram-Schmidt \cite{Sheldon:1997} first to orthogonalize the vectors and then normalize them. The first vector can be chosen  
\begin{subequations}
\begin{align}
    \hat{\boldsymbol{w}}_1&= 
    \hat{\boldsymbol{\rho}} , 
    \shortintertext{while the other one}
    \boldsymbol{w}_2 
    &= 
    \boldsymbol{\hat{a}}^{-1} 
    -
    (\boldsymbol{\hat{a}}^{-1} \cdot \hat{\boldsymbol{w}}_1)
    \hat{\boldsymbol{w}}_1
    =
    \boldsymbol{\hat{a}}^{-1} 
    -
    H
    \hat{\boldsymbol{w}}_1 
    \label{eq:w2_first}
    \, ,
\end{align}    
\end{subequations}
see that the norm of the last vector,
\begin{subequations}
\begin{align}
    \Vert\boldsymbol{w}_2\Vert^2
    &=
    \boldsymbol{w}_2 \cdot \boldsymbol{w}_2 
    =
    \left(\boldsymbol{\hat{a}}^{-1} 
    -
    H
    \hat{\boldsymbol{w}}_1
    \right)
    \cdot
    \left(
    \boldsymbol{\hat{a}}^{-1} 
    -
    H
    \hat{\boldsymbol{w}}_1
    \right) \, , \\
    &=
    \Vert\boldsymbol{\hat{a}}^{-1}\Vert^2
    -
    2
    H
    \boldsymbol{\hat{a}}^{-1}
    \cdot
    \hat{\boldsymbol{w}}_1
    +
    H^2
    \, , \\
    &=
    1-2H^2
    +H^2
    \, , \\
    &=
    1-\cos^2\theta = \sin^2\theta
    \, .
\end{align}
\end{subequations}
Two cases should be distinguished at this point. The first one corresponds to condition $\Vert\boldsymbol{w}_2\Vert=0$, which can only happen if $\theta=0$ (the possibility $\theta = \pi$ is ruled out since vectors $\bm{\rho}$ and $\bm{a}^{-1}$ cannot have negative entries). Condition $\theta=0$ is only met if the vector of densities, $\bm{\rho}$, and the vector of flexibilities, $\bm{a}^{-1}$, are parallel (see Eq.~\eqref{eq:angle}). In this improbable case, the solution to the optimization problem posed in Eq.~\eqref{eq:opt} would be directly given by $\hat{\bm{l}}=\hat{\bm{\rho}}=\bm{\hat{a}}^{-1}$. The second case is the general scenario where $\bm{\rho}$ and $\bm{a}^{-1}$ are not parallel vectors. Then, we can define $\hat{\boldsymbol{w}}_2 = \boldsymbol{w}_2 / \Vert\boldsymbol{w}_2\Vert$. From \cref{eq:w2_first} it easily follows that

\begin{align}
    \hat{\boldsymbol{w}}_2 
    =
    {\boldsymbol{\hat{a}}^{-1} \over \sin \theta}
    -
    { \hat{\bm{w}}_1
    \over
    \tan \theta} 
    \to 
    \boldsymbol{\hat{a}}^{-1}
    =
    \cos \theta
    \hat{\boldsymbol{w}}_1
    +
    \sin \theta
    \hat{\boldsymbol{w}}_2 \, .
\end{align}
Vector $\hat{\boldsymbol{l}}$ can be expressed in the orthonormal basis as $\hat{\boldsymbol{l}} = \alpha \hat{\boldsymbol{w}}_1 + \beta \hat{\boldsymbol{w}}_2$, 
where the projections $\alpha$ and $\beta$ (which we are going to maximize over) abide by $\alpha^2 + \beta^2 = 1$ (unit-vector condition).
Passing to compute the inner products,
\begin{align}
    \hat{\boldsymbol{l}} 
    \cdot
    \hat{\boldsymbol{\rho}}
    =
    (\alpha \hat{\boldsymbol{w}}_1 + \beta \hat{\boldsymbol{w}}_2)
    \cdot 
    \hat{\boldsymbol{w}}_1
    =
    \alpha \, ,
\end{align}
whereas
\begin{align}
    \hat{\boldsymbol{l}} 
    \cdot
    \boldsymbol{\hat{a}}^{-1}
    =
    (\alpha \hat{\boldsymbol{w}}_1 + \beta \hat{\boldsymbol{w}}_2)
    \cdot 
    (\cos \theta \hat{\boldsymbol{w}}_1 + \sin \theta \hat{\boldsymbol{w}}_2)
    =
    \alpha \cos \theta + \beta \sin \theta \, .
\end{align}

Thus, the problem of simultaneously maximizing the projection of $\hat{\boldsymbol{l}}$ over $\hat{\boldsymbol{\rho}}$ and $\boldsymbol{\hat{a}}^{-1}$ reduces to finding the maxima of a one-variable function,
\begin{align}
    F(\alpha)
    =
    \alpha (\alpha \cos \theta + \sqrt{1-\alpha^2} \sin \theta) \, ,
\end{align}
where $\theta$ is known since $\hat{\boldsymbol{\rho}}$ and $\boldsymbol{\hat{a}}^{-1}$ are defined (i.e. the materials and the layering order are defined, it remains to assign the thickness of each layer).

It can be verified that $F'(\alpha)=0$ when $\alpha = \cos (\theta/2)$. Thus, $\hat{\boldsymbol{l}} 
= 
\cos (\theta/2) \hat{\boldsymbol{w}}_1 
+ 
\sin (\theta/2) \hat{\boldsymbol{w}}_2$. In other words, in order to maximize the product $(\hat{\boldsymbol{l}} \cdot \hat{\boldsymbol{\rho}})
    (\hat{\boldsymbol{l}} \cdot \boldsymbol{\hat{a}}^{-1})$, vector $\hat{\boldsymbol{l}}$ needs to be placed in the plane spanned by $\hat{\boldsymbol{\rho}}$ and $\boldsymbol{\hat{a}}^{-1}$, exactly in between them.

Geometrically, this means that the optimal direction, $\hat{\boldsymbol{l}}^*$, corresponds to a unitary vector parallel to the mean vector defined by $\hat{\boldsymbol{\rho}}$ and $\boldsymbol{\hat{a}}^{-1}$, as represented in Fig.~\ref{fig:sphere} for a case with 3 layers:
\begin{align}
\label{eq_analyt_opt_b}
    \hat{\boldsymbol{l}}^*
    =
 \frac{  
    \hat{\boldsymbol{\rho}}
    +
    \boldsymbol{\hat{a}}^{-1}}
    {\| \hat{\boldsymbol{\rho}}
    +
    \boldsymbol{\hat{a}}^{-1}\|},
\end{align}
or, equivalently,
\begin{align}
\label{eq_analyt_opt}
    \hat{\boldsymbol{l}}^*
    \parallel
    {1 \over 2}
    \left(
    \hat{\boldsymbol{\rho}}
    +
    \boldsymbol{\hat{a}}^{-1}
    \right) \, .
\end{align}

\end{proof}

\begin{figure}[H]
    \centering
    \includegraphics[width=0.45\textwidth]{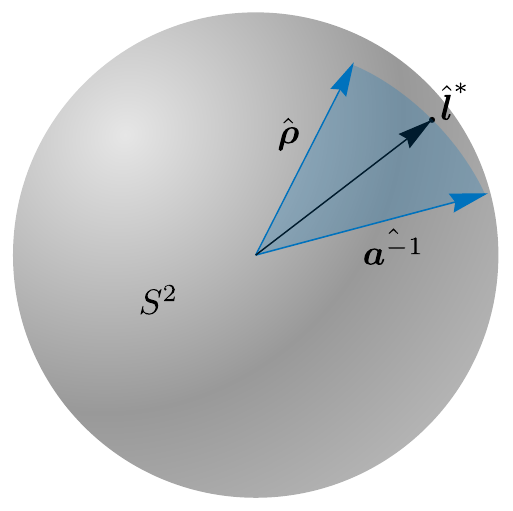}
    \caption{Schematic representation of the solution to the optimization problem given by Eq.~\eqref{eq:opt} (minimization of the curvature of the half-trace function at the origin), for a case with 3 layers.}
    \label{fig:sphere}
\end{figure}

\section{Numerical Verification of optima}
\label{sec:numerical_verif}

The analytical developments of the last section can be summarized as follows: 
given the material properties and layering order for a phononic crystal formed by the periodic repetition of a layered cell, we are interested in finding the set of layer thicknesses that minimizes the first cut-off frequency of the periodically-laminated medium. 
An analytical solution for this optimization problem, based on a second-order Taylor approximation, has been found (Eq.~\eqref{eq_analyt_opt}). 
The present section is intended to verify this analytical solution, by comparing it with a numerically obtained optimal layering.

Three different cases are considered, for which the material properties of the different layers are given in Tables~\ref{table:Case 1}-\ref{table:Case 3}. In each of these scenarios, a numerical optimum has been obtained by using a \textit{surrogate optimization} algorithm in Matlab \cite{Matlab}, which searches for the global minimum of certain objective function. In this case, the objective function is the first cut-off frequency of the half-trace function, $\eta(\omega)$. Note that this approach does not use any Taylor approximations, but the actual half-trace function given by Eq.~\eqref{eq:eta_omega}. As discussed in Section~\ref{sec:curvature}, the optimization problem needs to include a restriction on the norm of the vector of thicknesses: $\Vert \bm{l}\Vert\leq\Vert \bm{l}\Vert_{max}$. For the three presented examples, the maximum norm has been set to \mbox{$\Vert \bm{l}\Vert_{max}=5$ cm}. 

In addition to the numerical and analytical solutions of the optimization problem, a random set of thicknesses has been added to the results for the sake of comparison. This layering has been obtained by assigning to each layer thickness a random variable with uniform distribution between 0 and 1 cm, and then scaling the resulting vector of thicknesses to make $\Vert \bm{l}\Vert=\Vert \bm{l}\Vert_{max}=5$ cm.

Hence, three different half-trace functions are obtained for each of the three cases: the numerical optimum, the analytical optimum and the one corresponding to the random layering. These curves are jointly represented in Fig.~\ref{fig:half-trace functions}. Note that the second-order Taylor approximations of all the half-trace functions are also shown in the graphs (dashed lines). 

\begin{figure}
    \centering
    \includegraphics[width=1\textwidth]{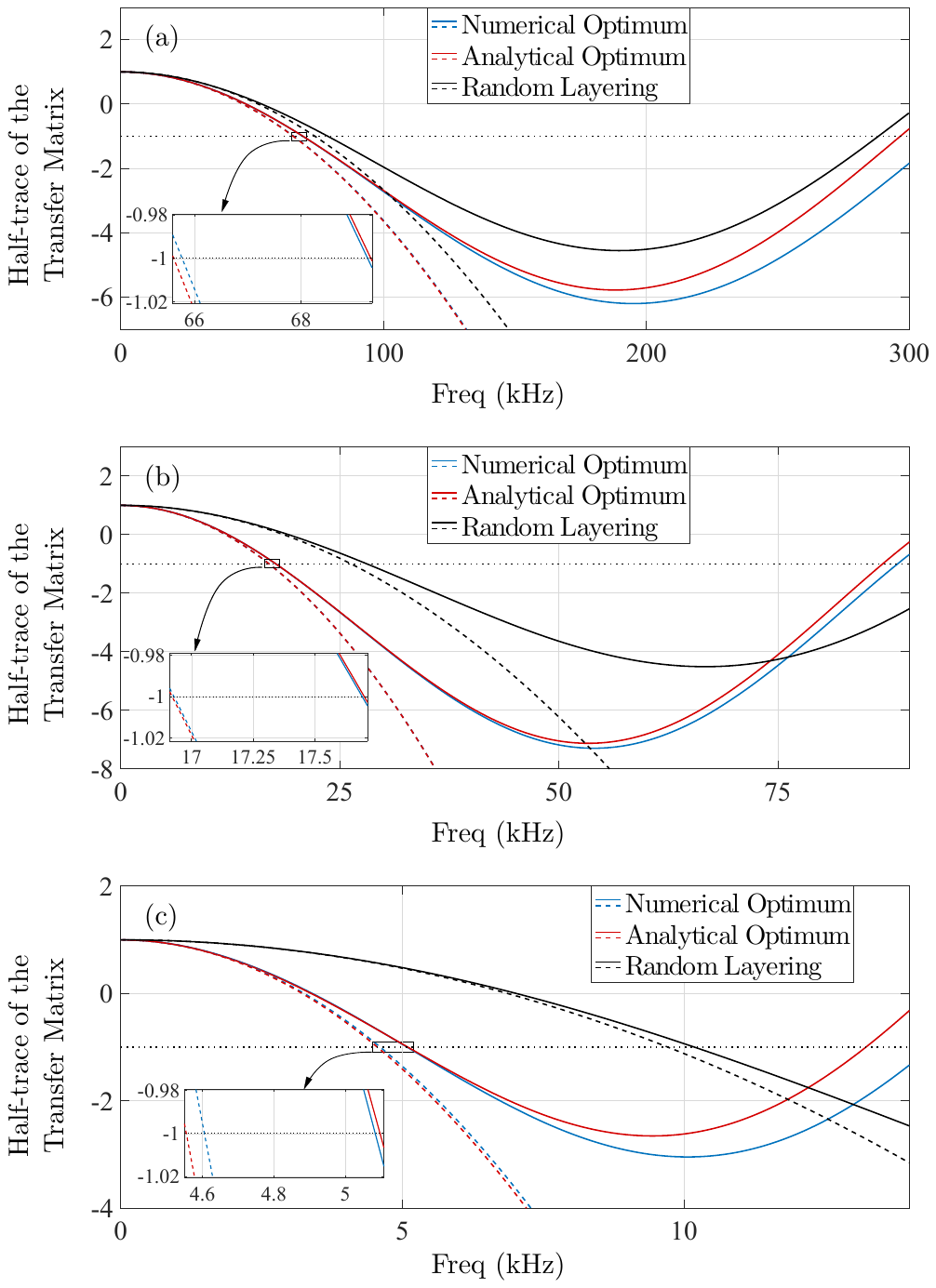}
    \caption{Representation of the half-trace of the transfer matrix against frequency for (a) Case 1, (b) Case 2 and (c) Case 3. Three different layerings are considered: the analytical optimum given by Eq.~\eqref{eq_analyt_opt}, the numerical optimum and a random layering used as a reference. A solid line represents the exact half-trace function, while a dashed line represents the second-order Taylor approximation at the origin.}
    \label{fig:half-trace functions}
\end{figure}

\begin{figure}
    \centering
    \includegraphics[width=0.7609\textwidth]{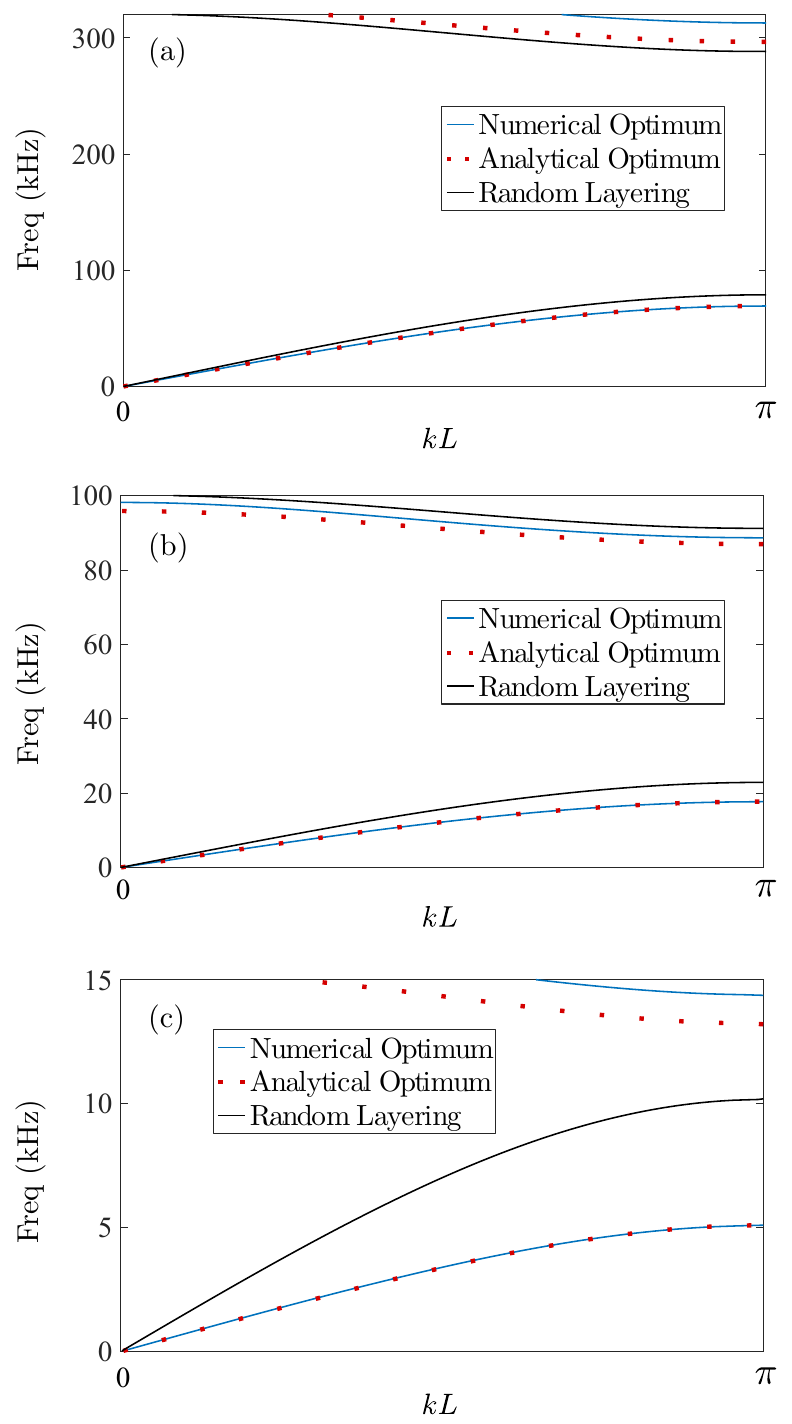}
    \caption{Dispersion diagrams for (a) Case 1, (b) Case 2 and (c) Case 3. The angle $kL$ in the horizontal axes, where $k$ is the wavenumber and $L$ is the total thickness of the repetitive crystal cell, is related to the half-trace function, $\eta(\omega)$, through Eq.~\eqref{eq_bloch}. Three different layerings are considered: the analytical optimum given by Eq.~\eqref{eq_analyt_opt}, the numerical optimum and a random layering used as a reference.}
    \label{fig:dispersion diagrams}
\end{figure}

The following observations can be extracted from Fig.~\ref{fig:half-trace functions}:
\begin{itemize}
    \item [\tiny\textbullet] Both the numerical and analytical optima yield a significant improvement (reduction in the first cut-off frequency) with respect to the random layering (particularly in Cases 2 and 3).
    \item [\tiny\textbullet] The first cut-off frequency provided by the analytical and numerical optimization procedures are very similar, with a slight advantage on the side of the numerical solution. This was expected, since the analytical solution relies upon a Taylor approximation and the numerical optimum does not.
   \item [\tiny\textbullet] Note the details contained in the inset of \Cref{fig:half-trace functions}. 
   If we restrict our attention to the dashed parabolic curves (second-order Taylor approximations of the half-trace functions), the design based on \Cref{eq_analyt_opt} is observed to reach below -1 faster than the quadratic approximation of the numerical optimization. This is also an expected result, considering that the analytical solution minimizes the curvature of the half-trace function at the origin while the numerical optimization is directly tasked with minimizing the cut-off frequency.
\end{itemize}

The information contained in Fig.~\ref{fig:half-trace functions} has also been represented in the form of dispersion diagrams in Fig.~\ref{fig:dispersion diagrams} (in this case the second-order Taylor approximations are not shown). The angle $kL$ in the horizontal axes of Fig.~\ref{fig:dispersion diagrams}, where $k$ is the wavenumber and $L$ is the total thickness of the repetitive crystal cell, is related to the half-trace function, $\eta(\omega)$, through relation
\begin{equation}
\label{eq_bloch}
    \cos(kL)=\eta(\omega),
\end{equation}
which is but a re-statement of Bloch's periodicity condition (see Ref.~\cite{JMPS} for details).

The results presented in this section show that, in the three considered scenarios, the analytical optimum given by Eq.~\eqref{eq_analyt_opt} constitutes a reasonable approximation to the optimal layering obtained with a global numerical optimization procedure, which confirms that minimizing the curvature of the half-trace function at the origin can be a powerful strategy in the design of phononic crystals for minimal cut-off frequency.

To finish this section, we stress that the analysis was predicated on both the materials and the order of the layers being predefined. 
Removing these assumptions would provide extra degrees of freedom to designers, but we are still figuring out how to tackle the problem with a relaxed set of assumptions. 
Hence, open questions to be addressed in future work can be: what is the optimal ordering of layers like? or what is the marginal improvement of adding extra layers to the unit cell?  

\begin{table}[h!]
\centering
\caption{Parameter values for Case 1. Three different layerings are provided: the numerical optimum ($l_i^{num}$), the analytical optimum obtained with Eq.~\eqref{eq_analyt_opt} ($l_i^{an}$) and a random layering ($l_i^{rand}$).}
\label{table:Case 1}
\begin{tabular}{rrrrrr}
Layer & $\rho_i$ (kg/m) & $\quad a_i$ (N) & $l_i^{num}$ (cm) & $l_i^{an}$ (cm) & $l_i^{rand}$ (cm) \\ \hline
1     & 31              & $30\cdot10^9$   & 1.76             & 2.03            & 3.31              \\
2     & 2.9             & $4\cdot10^9$    & 3.52             & 3.37            & 3.52              \\
3     & 55              & $50\cdot10^9$   & 3.08             & 3.09            & 1.29             
\end{tabular}
\end{table}

\begin{table}[h!]
\centering
\caption{Parameter values for Case 2. Three different layerings are provided: the numerical optimum ($l_i^{num}$), the analytical optimum obtained with Eq.~\eqref{eq_analyt_opt} ($l_i^{an}$) and a random layering ($l_i^{rand}$).}
\label{table:Case 2}
\begin{tabular}{rrrrrr}
Layer & $\rho_i$ (kg/m) & $\quad a_i$ (N) & $l_i^{num}$ (cm) & $l_i^{an}$ (cm) & $l_i^{rand}$ (cm) \\ \hline
1     & 100              & $8\cdot10^9$   & 3.50             & 3.46            & 4.66              \\
2     & 5             & $0.5\cdot10^9$    & 2.90             & 2.85            & 1.78              \\
3     & 9              & $0.7\cdot10^9$   & 2.09             & 2.22            & 1.64             
\end{tabular}
\end{table}

\begin{table}[h!]
\centering
\caption{Parameter values for Case 3. Three different layerings are provided: the numerical optimum ($l_i^{num}$), the analytical optimum obtained with Eq.~\eqref{eq_analyt_opt} ($l_i^{an}$) and a random layering ($l_i^{rand}$).}
\label{table:Case 3}
\begin{tabular}{rrrrrr}
Layer & $\rho_i$ (kg/m) & $\quad a_i$ (N) & $l_i^{num}$ (cm) & $l_i^{an}$ (cm) & $l_i^{rand}$ (cm) \\ \hline
1     & 267              & $2.4\cdot10^9$   & 3.46             & 3.09            & 2.80              \\
2     & 5.4             & $2.2\cdot10^9$    & 0.43             & 0.19            & 2.36              \\
3     & 11.8              & $0.3\cdot10^9$   & 0.84             & 1.09            & 1.08             \\
4     & 5.3              & $0.5\cdot10^9$   & 0.83             & 0.63            & 3.23             \\
5     & 76              & $0.1\cdot10^9$   & 3.38             & 3.72            & 0.12             
\end{tabular}
\end{table}

\section{Does this optimal layering also maximize the bandgap width?}
\label{sec:broadband}

One may hope that designing a layering to minimize the frequency for which the first bandgap opens (i.e., $\omega$ s.t. $\eta(\omega)=-1$) also serves to create a broadband gap (another design goal addressed in the literature \cite{broadband:2015,Shmuel:2016,Lustig:2018,Morini:2019,Li:2021}). 
This supposition hinges on the expectation that making the half-trace function to go below $-1$ as ``fast'' as possible also means that it stays below $-1$ for a wide interval of frequencies 

We show that this is generally not the case by means of a counter-example. 
Resorting to a particular case presented in Ref.~\cite{JMPS} (code ``IWTH27''), we pursue the optimization for that case (four layers), both numerically and analytically based on the curvature criterion, imposing the norm $\Vert\bm{l}\Vert$ to be the same as in the IWTH27-layering.
We find, \Cref{fig:IWTH27}, that the original laminate does present a much wider gap which opens at a frequency substantially greater than the one found by our optimization (which again agrees well with the numerically-optimized one). 

Yet another proof is provided by the simplest case, two-layer laminate, for which we know exactly (thanks to the toroidal representation) that the optimal thicknesses of the phases in the unit cell must satisfy $l_2 / l_1 = c_2^2 / c_1^2$ \cite{Shmuel:2016}. 
Conversely, our criterion based on curvature, \cref{eq_analyt_opt_thm}, yields $l_2 / l_1 = \mathcal{K} \cdot c_2^2 / c_1^2$, in which the extra factor takes a complex form dependent on the mechanical properties of the laminates
\begin{align}
    \mathcal{K}
    =
    { \sqrt{(\rho_2/\rho_1)^2 + (c_1/c_2)^2} + \sqrt{(a_1/a_2)^2 + (c_2/c_1)^2}
    \over 
    \sqrt{(\rho_1/\rho_2)^2 + (c_2/c_1)^2} + \sqrt{(a_2/a_1)^2 + (c_1/c_2)^2}
    } \, ,
\end{align}
which is equal to one if the mechanical properties of the phases are equal too.

We venture that a more suitable design criterion for maximum broadband attenuation should be minimizing the value of $\eta(\omega)$ for the first frequency, such that $\eta'(\omega)=0$. 
%
An analytical strategy to ensure a balanced objective (low cut-off frequency while ensuring enough bandwidth) has not been devised yet. 

\begin{figure}[H]
    \centering
    \includegraphics[width=1\textwidth]{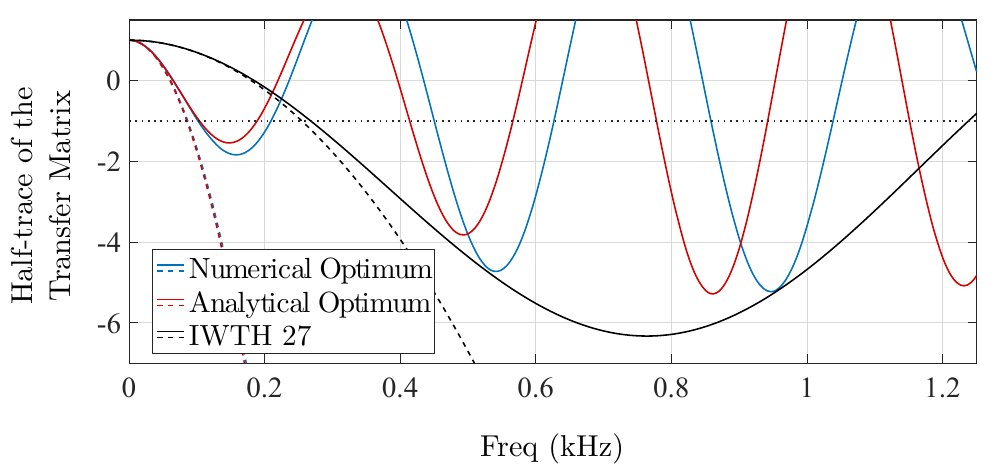}
    \caption{Representation of the half-trace function against frequency for a phononic crystal mirroring the mechanical properties of Kik-Net Site IWTH27 (Ref.~\cite{JMPS}), together with the corresponding analytical and numerical optima. A solid line represents the exact half-trace function, while a dashed line represents the second-order Taylor approximation at the origin.}
    \label{fig:IWTH27}
\end{figure}

\section{Concluding remarks}
\label{sec:final}

We have presented a criterion to minimize the frequency at which the first frequency bandgap opens in 1D phononic crystals. 
Our approach stemmed from the general closed-form expression of cumulative transfer matrices, for which we have also provided a new proof, and is based on minimizing the curvature of the half-trace function at zero frequency. 
A relevant feature of the proposed strategy is that the optimum can be expressed in an explicit analytical manner. Moreover, our analytical approach has been validated via comparison to numerical optimization: in all cases under scrutiny, the analytical result is found to provide an accurate approximation to the numerical one. 
%
It has also been shown that designs that minimize this cut-off frequency do not necessarily yield large bandwidth gaps. 
Future work should focus on devising strategies rooted in the harmonic decomposition that may balance, on one hand, the attainment of broadband attenuation and, on the other hand, doing so at low frequencies. 
Moreover, this should also be tackled under a relaxed set of assumptions, e.g., without pre-defining the ordering of the layers.

The presented results are proof of the deep insights enabled by the adoption of transfer matrices as the building block of the analysis of phononic crystals. 
Such an approach can also help to elucidate more complex wave propagation phenomena, including inclined waves, systems that feature mode conversion and surface waves \cite{JMPS}. 


We have demonstrated one of the practical uses of the transfer matrix mathematical formalism, but we foresee that there will be many more: from designing analog computers \cite{analog:2021} to the aforementioned broadband gap design (\Cref{sec:broadband}) for vibration isolation.

Finally, let us stress that our results are directly translatable to other fields as, per instance, photonic crystals \cite{Kushwaha:1993}.



\section*{Authorship contribution statement}

\noindent JGC: Conceptualization, Formal analysis, Writing – original draft, Numerical validation.\\

\noindent ML: Conceptualization, Mathematical proofs, Formal analysis, Writing – reviewing and editing.\\

\noindent JGS: Conceptualization, Formal analysis, Writing – original draft.

\section*{Data availability}

The Matlab scripts used to generate the numerical results presented  in this work can be obtained via correspondence to the first author.

\bibliography{bibliography}
\bibliographystyle{unsrt}

\end{document}